\documentclass[conference]{IEEEtran}
\IEEEoverridecommandlockouts

\usepackage{amsmath,amsfonts}
\usepackage{array}
\usepackage{textcomp}
\usepackage{stfloats}
\usepackage{verbatim}
\usepackage{graphicx}
\usepackage{cite}
\usepackage{booktabs}
\usepackage{diagbox}
\allowdisplaybreaks[4]

\usepackage{algorithm,algorithmic}
\usepackage{subfigure}
\usepackage{float}
\usepackage{amssymb,amsfonts}
\usepackage{amsmath}
\usepackage{amsthm}
\theoremstyle{definition}

\newtheorem{proposition}{Proposition}

\newtheorem{remark}{Remark}

\usepackage{color}
\usepackage{graphicx}
\usepackage{multirow}
\usepackage{epsfig}
\usepackage{epstopdf} 
\usepackage{extarrows}
\usepackage{tikz}
\usetikzlibrary{arrows.meta}
\usepackage{setspace}
\usepackage{colortbl}

\begin{document}

\title{
Joint Channel Sounding and Source-Channel Coding \\for MIMO-OFDM Systems: Deep Unified Encoding and Parallel Flow-Matching Decoding
}

\author{\IEEEauthorblockN{Hao Jiang\IEEEauthorrefmark{1}, Xiaojun Yuan\IEEEauthorrefmark{1}, and Qinghua Guo\IEEEauthorrefmark{2}}
\IEEEauthorblockA{\IEEEauthorrefmark{1}{National Key Laboratory of Wireless Communications, Univ. of Elec. Sci. and Tech. of China, Chengdu, China}}
\IEEEauthorblockA{\IEEEauthorrefmark{2}{
School of Engineering, Univ. of Wollongong, Wollongong, Australia}\\
Email: jh@std.uestc.edu.cn,
  xjyuan@uestc.edu.cn,
  qguo@uow.edu.au}
}

\maketitle

\begin{abstract}
In this work, we propose a deep unified (DU) encoder that embeds source information in a codeword that contains sufficient redundancy to handle both channel and source uncertainties, without enforcing an explicit pilot-data separation. At the receiver, we design a parallel flow-matching (PFM) decoder that leverages flow-based generative priors to jointly estimate the channel and the source, yielding much more efficient inference than the existing diffusion-based approaches. To benchmark performance limits, we derive the Bayesian Cram\'er-Rao bound (BCRB) for the joint channel and source estimation problem. Extensive simulations over block-fading MIMO-OFDM channels demonstrate that the proposed DU-PFM approach drastically outperforms the state-of-the-art methods in both channel estimation accuracy and source reconstruction quality.
\end{abstract}



\section{Introduction}
\label{sec:introduction}

In semantic communication, DNN-based joint source-channel coding (DJSCC) \cite{wu2023fusion, wu2024deep, gunduz2024joint, yang2024swinjscc, xu2025tdjscc} has become a prevailing approach to the design of neural codecs for multi-modal source transmission, in which source data (e.g., audio, images, and video) are directly mapped into channel inputs. This approach has proven successful in various communication systems, encompassing multi-user transmission \cite{wu2023fusion}, multiple-input multiple-output (MIMO) systems \cite{wu2024deep}, and orthogonal frequency-division multiplexing (OFDM) \cite{xu2025tdjscc}, consistently exhibiting significantly improved spectrum efficiency. Yet, most existing work on DJSCC either assumes that the channel is \textit{a priori} known, or takes some conventional channel sounding techniques, such as sending orthogonal pilots at the transmitter and using canonical channel estimation methods at the receiver \cite{wu2024deep, yang2024swinjscc}. Then, a natural question arises: How much overhead do we need to pay for the acquisition of channel state information (CSI) in a DJSCC-based communication system?

In this paper, we attempt to answer this question by integrating DJSCC into MIMO-OFDM systems defined in modern standards such as 3GPP 5G-NR \cite{3gpp.38.201} and WiFi-7 \cite{deng2020ieee}. In a practical MIMO-OFDM system, the channel typically varies fast over frequency due to the multipath effect and over time due to user mobility, which implies that the amount of time-frequency resource required in CSI acquisition is unignorable. In current commercial MIMO-OFDM systems, approximately 15\%-20\% of the time-frequency resources are dedicated to pilot transmission for channel sounding \cite{tr385g, rusek2012scaling}. In fact, pilot overhead can be reduced by exploiting the inherent low-dimensionality of the MIMO-OFDM channel \cite{kuai2019structured}. 
The question is then what is the best design practice that can resolve the channel uncertainty at the receiver by introducing a minimal amount of redundancy at the transmitter. From a more general perspective, the receiver needs to resolve both channel and source uncertainties based on the transmission redundancy introduced at the transmitter side. This implies that channel sounding and source-channel coding can be unified to enable a more efficient allocation of resources, thereby achieving a higher spectral efficiency of the system.

Building upon this insight, we take a joint channel sounding and source-channel coding approach to reduce CSI acquisition overhead and improve spectral efficiency by exploiting the powerful encoding and decoding capabilities of DNN-based techniques. 
Specifically, we propose a deep unified (DU) encoder at the transmitter that maps source data into channel input symbols via a DNN. It is worth noting that a codeword of the DU encoder does not necessarily exhibit a clear structural separation between the redundant component and the information-bearing component, although conventional designs with explicit redundancy, i.e., pilots \cite{mashhadi2021pruning, han2025interference}, are included as special cases of our proposed DU design. 
At the receiver, we propose a powerful parallel flow-matching (PFM) decoder to jointly estimate the channel and source. In PFM, we leverage the flow model \cite{lipman2023flow}, which effectively characterizes high-dimensional data posing low-dimensional structures, to capture the prior information of the channel and the source data.
This allows the decoder to perform  Bayesian inference for joint channel estimation and source reconstruction.
Simulations demonstrate that the proposed DU-PFM approach drastically outperforms existing representative methods in terms of spectral efficiency, channel estimation accuracy, and perceptual quality of recovered source data.
The contributions of this work are summarized as follows:
\begin{itemize}
    \item To the best of our knowledge, this work presents the first systematic study of the joint channel sounding and source-channel coding problem under limited time-frequency resources in MIMO-OFDM scenarios.
    \item We propose a deep unified encoder that implicitly introduces redundancy for channel sounding and embeds source information into channel input symbols, together with a powerful parallel flow-matching decoder for simultaneous channel estimation and source reconstruction.

    \item We bridge generative learning and estimation theory to derive the BCRB for joint channel estimation and source recovery, where the learned flow models are used to approximate the prior Fisher information matrices of the channel and transmit signal. 
    
    \item Through extensive simulations under 3GPP UMa-NLOS channel conditions, DU-PFM exhibits markedly superior channel estimation accuracy and source reconstruction quality relative to the state-of-the-art approaches.
\end{itemize}

\section{System Model and Problem Formulation}\label{Section1}

\subsection{System Model}
We consider a block-fading multi-user MIMO-OFDM system, where $N_t$ single-antenna transmitters communicate with a receiver equipped with $N_r$ antennas over $N_f$ adjacent subcarriers. The channel is assumed to remain constant over a transmission block of $T$ OFDM symbols and vary independently over blocks.
Let $h_{f,k,n} \in \mathbb{C}$ denote the frequency-space-domain channel coefficient from the $k$-th transmitter to the $n$-th receive antenna on the $f$-th subcarrier. 
The uplink channel can be denoted as a tensor
$\boldsymbol{\mathcal{H}} \in \mathbb{C}^{N_f\times N_t\times N_r}$ with the $(f,k,n)$-th element being $h_{f,k,n}$.
We next define the transmit signal in a transmission block as a tensor $\boldsymbol{\mathcal{X}} \in \mathbb{C}^{N_f\times T\times N_t}$, where the $(f,t,k)$-th element, i.e., $x_{f,t,k}$, denotes the symbol transmitted by the $k$-th transmitter on the $f$-th subcarrier during the $t$-th OFDM symbol. Each transmit antenna is subject to an average power constraint $P$ over the $N_fT$ transmit symbols.
Assume that the cyclic prefix (CP) is longer than the channel delay spread. Then, after CP removal, the signal received by the $n$-th receive antenna on the $f$-th subcarrier during the $t$-th OFDM symbol is given by
$y_{f,t,n}=\sum\nolimits^{N_t}_{k=1} x_{f,t,k} h_{f,k,n} + w_{f,t,n}$,
where $w_{f,t,n}\sim \mathcal{CN}(0,\sigma^2_n)$ denotes the $(f,t,n)$-th element of the additive white Gaussian noise $\boldsymbol{\mathcal{W}}\in \mathbb{C}^{N_f \times T\times N_r}$. For subcarrier $f$, the receive signal over a transmission block is given by 
\begin{align}\label{systemmodel0}
    \boldsymbol{\mathcal{Y}}_{f,:,:}=\boldsymbol{\mathcal{X}}_{f,:,:} \boldsymbol{\mathcal{H}}_{f,:,:} + \boldsymbol{\mathcal{W}}_{f,:,:} \in \mathbb{C}^{T\times N_r},
\end{align}
where $\boldsymbol{\mathcal{Y}} \in \mathbb{C}^{N_f \times T\times N_r}$ is the receive signal tensor.


We propose a DU encoder to learn a compact representation that implicitly incorporates redundancy to combat channel uncertainty while preserving source information. Let $\boldsymbol{\mathcal{S}}_{k}$ be the source data to be transmitted by the $k$-th transmitter. The DU encoder of the $k$-th transmitter, denoted as $f_{\boldsymbol{\gamma}_k}$ parameterized by $\boldsymbol{\gamma}_k$,
maps the source data $\boldsymbol{\mathcal{S}}_{k}$ into a codeword for transmission as
\begin{equation}\label{eq:encoder_PF}
    \boldsymbol{\mathcal{X}}_{:,:,k} = f_{\boldsymbol{\gamma}_k}(\boldsymbol{\mathcal{S}}_{k})\in\mathbb{C}^{N_f\times T_{\mathcal{S}}},
\end{equation}
where $T_{\mathcal{S}}$ denotes the total number of OFDM symbols needed for transmitting $\boldsymbol{\mathcal{S}}_{k}$. Existing schemes utilizing explicit redundancy, i.e., pilot symbols \cite{mashhadi2021pruning, han2025interference}, to combat channel uncertainty can be incorporated into the proposed DU design as special cases.
At the receiver side, we propose a DNN-based decoder $g_{\boldsymbol{\varTheta}}$, parameterized by $\boldsymbol{\varTheta}$, which jointly estimates both the channel and the source data from the receive signal as
\begin{align}\label{receiver}
    \{\hat{\boldsymbol{\mathcal{H}}}, {\{\hat{\boldsymbol{\mathcal{S}}}_k\}^{N_t}_{k=1}}\} = g_{\boldsymbol{\varTheta}}\left(\boldsymbol{\mathcal{Y}}\right).
\end{align}

\subsection{Problem Description}
To enable joint channel and source estimation, our system design follows the mutual information maximization principle \cite{tse2005fundamentals, cai2025end}. Specifically, we aim to maximize the mutual information between $\{\boldsymbol{\mathcal{H}}, {\{\boldsymbol{\mathcal{S}}_k\}^{N_t}_{k=1}}\}$ and $\boldsymbol{\mathcal{Y}}$, i.e.,
$I(\boldsymbol{\mathcal{H}},{\{\boldsymbol{\mathcal{S}}_k\}^{N_t}_{k=1}};\boldsymbol{\mathcal{Y}})=H(\boldsymbol{\mathcal{H}},{\{\boldsymbol{\mathcal{S}}_k\}^{N_t}_{k=1}}) - H(\boldsymbol{\mathcal{H}},{\{\boldsymbol{\mathcal{S}}_k\}^{N_t}_{k=1}}\vert\boldsymbol{\mathcal{Y}})$,
where $H(\boldsymbol{\mathcal{H}}, {\{\boldsymbol{\mathcal{S}}_k\}^{N_t}_{k=1}})$ remains constant for a given system. Therefore, the maximization of the mutual information is equivalent to the minimization of 
$H(\boldsymbol{\mathcal{H}},{\{\boldsymbol{\mathcal{S}}_k\}^{N_t}_{k=1}}\vert  \boldsymbol{\mathcal{Y}}) = \mathbb{E}_{p(\boldsymbol{\mathcal{Y}}, \boldsymbol{\mathcal{H}}, {\{\boldsymbol{\mathcal{S}}_k\}^{N_t}_{k=1}})}[-\ln p(\boldsymbol{\mathcal{H}},{\{\boldsymbol{\mathcal{S}}_k\}^{N_t}_{k=1}}\vert  \boldsymbol{\mathcal{Y}})]$, in which the posterior distribution 
is generally intractable due to the unknown prior distributions of the channel and the source data. To avoid this problem, a widely used approach is to replace the posterior distribution with a variational distribution $q_{\boldsymbol{\varTheta}}(\boldsymbol{\mathcal{H}},{\{\boldsymbol{\mathcal{S}}_k\}^{N_t}_{k=1}}\vert\boldsymbol{\mathcal{Y}})$ parameterized by $\boldsymbol{\varTheta}$ \cite{cai2025end, jiang2025blind}. 
This yields the variational upper bound of the conditional entropy, i.e., $\mathbb{E}_{p(\boldsymbol{\mathcal{Y}}, \boldsymbol{\mathcal{H}}, {\{\boldsymbol{\mathcal{S}}_k\}^{N_t}_{k=1}})}[-\ln q_{\boldsymbol{\varTheta}}(\boldsymbol{\mathcal{H}},{\{\boldsymbol{\mathcal{S}}_k\}^{N_t}_{k=1}}\vert  \boldsymbol{\mathcal{Y}})]$.
Thus, the mutual information maximization is relaxed to the minimization of the variational upper bound as
\begin{subequations}\label{optimizeP1}
\begin{align}
    \min_{\boldsymbol{\varGamma}, \boldsymbol{\varTheta}}& \,\,\mathbb{E}_
    {p(\boldsymbol{\mathcal{Y}},\boldsymbol{\mathcal{H}},{\{\boldsymbol{\mathcal{S}}_k\}^{N_t}_{k=1}})}
    [-\ln q_{\boldsymbol{\varTheta}}(\boldsymbol{\mathcal{H}},{\{\boldsymbol{\mathcal{S}}_k\}^{N_t}_{k=1}}\vert \boldsymbol{\mathcal{Y}})],\\
    \text{s. t.}& \,\, \|  f_{\boldsymbol{\gamma}_k}(\boldsymbol{\mathcal{S}}_{k})\|^2_F \leq N_fT_{\mathcal{S}}P, \,\,\forall k,\label{constraintX}
\end{align}
\end{subequations}
where $\boldsymbol{\varGamma}=[\boldsymbol{\gamma}_1,\cdots,\boldsymbol{\gamma}_{N_t}]$ and $\boldsymbol{\varTheta}$ are the parameters of 
the encoders and the decoder, respectively.
As the encoders and the decoder are trained jointly in \eqref{optimizeP1}, the channel uncertainty and the source uncertainty are coupled, resulting in an ill-conditioned nonconvex optimization problem. To avoid this, we employ a separate transceiver design, in which the codec is trained separately, as elaborated in the following subsections.

\subsection{Transceiver Design}
An illustration of the proposed transceiver design is provided in Fig.~\ref{fig:framework}, as detailed below.
\begin{figure}[htb]
    \centering
    \includegraphics[width=0.99\linewidth]{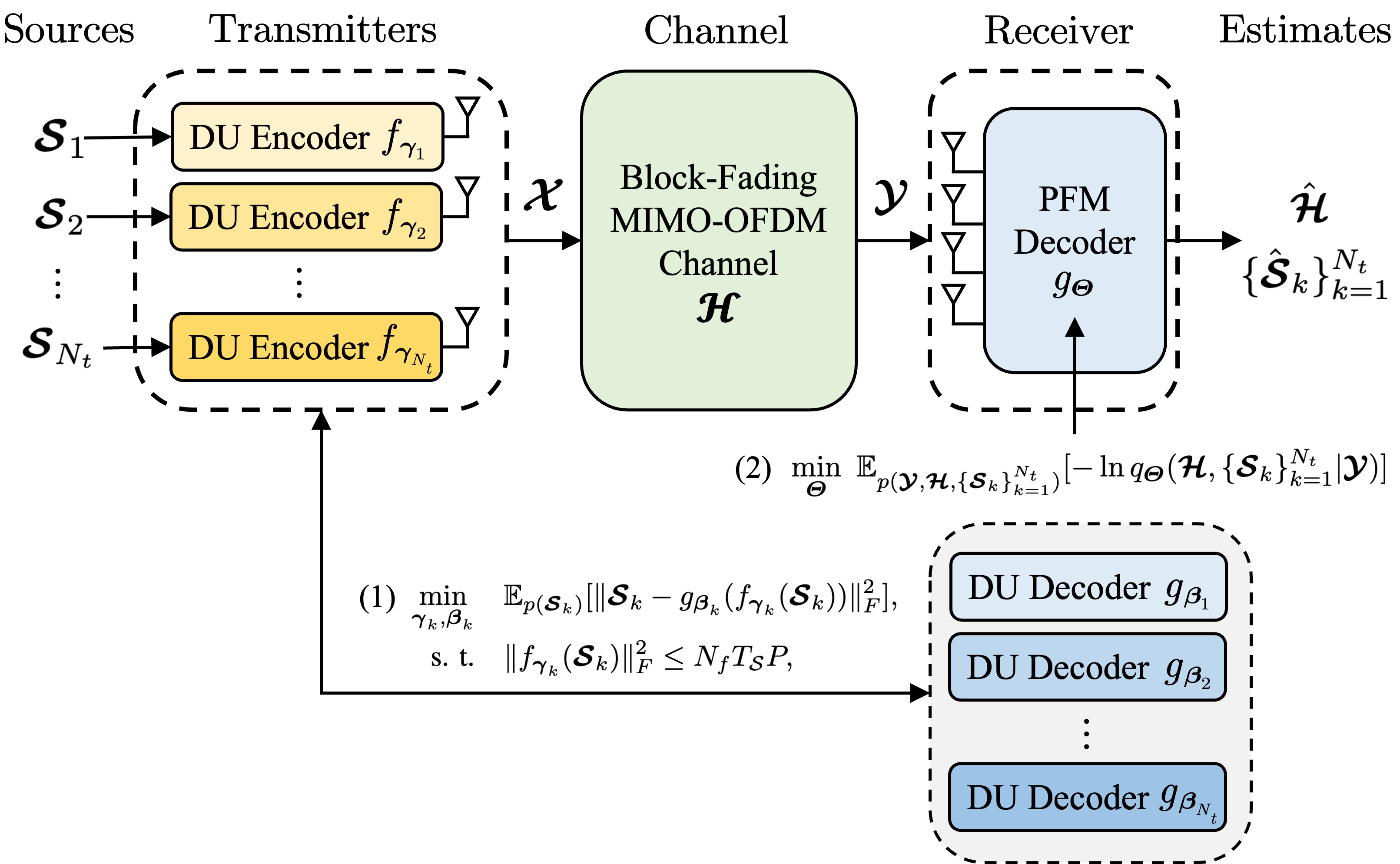}
    \caption{The proposed transceiver design, where $\{g_{\boldsymbol{\beta}_k}\}_{k=1}^{N_t}$ are the auxiliary DU decoders to facilitate the training of the DU encoders $\{f_{\boldsymbol{\gamma}_k}\}_{k=1}^{N_t}$, and $g_{\boldsymbol{\varTheta}}$ is the desired PFM decoder that achieves joint channel and source estimation.}
    \label{fig:framework}
\end{figure}
\subsubsection{Transmitter Design}
For the DU encoder design, as the channel is assumed to be unknown to the transmitters, we adopt a training approach that does not require the channel. Specifically, for any $k\in[1,\cdots,N_t]$, we train a DU codec pair $\{f_{\boldsymbol{\gamma}_k}, g_{\boldsymbol{\beta}_k}\}$ via the following optimization problem:
\begin{subequations}\label{eq:DUcodec}
\begin{align}
    \min_{\boldsymbol{\gamma}_k, \boldsymbol{\beta}_k}& \quad \mathbb{E}_{p(\boldsymbol{\mathcal{S}}_{k})}[\| \boldsymbol{\mathcal{S}}_{k} - g_{\boldsymbol{\beta}_k}(f_{\boldsymbol{\gamma}_k}(\boldsymbol{\mathcal{S}}_{k})) \|^2_F],\\ \text{s. t.}& \quad \| f_{\boldsymbol{\gamma}_k}(\boldsymbol{\mathcal{S}}_{k}) \|^2_F \leq N_fT_{\mathcal{S}}P,
\end{align}
\end{subequations}
where $g_{\boldsymbol{\beta}_k}$ is the DU decoder corresponding to the DU encoder $f_{\boldsymbol{\gamma}_k}$, and the constraint is the power budget in \eqref{optimizeP1}. 
Notably, $\{g_{\boldsymbol{\beta}_k}\}_{k=1}^{N_t}$ 
serve solely as auxiliary decoders to facilitate the training of DU encoders.
We adopt the Swin Transformer \cite{liu2021swin, yang2024swinjscc} as the backbone of our DU codec. 

\subsubsection{Receiver Design}
With the DU encoders $\{f_{\boldsymbol{\gamma}_k}\}^{N_t}_{k=1}$ fixed, the optimization problem in \eqref{optimizeP1} reduces to \begin{align}\label{optimizeP4}
    \min_{\boldsymbol{\varTheta}} \,\, \mathbb{E}_{p(\boldsymbol{\mathcal{Y}},\boldsymbol{\mathcal{H}},{\{\boldsymbol{\mathcal{S}}_k\}^{N_t}_{k=1}})}
    [-\ln q_{\boldsymbol{\varTheta}}(\boldsymbol{\mathcal{H}},{\{\boldsymbol{\mathcal{S}}_k\}^{N_t}_{k=1}}\vert \boldsymbol{\mathcal{Y}})].
\end{align}  
We factorize the variational posterior distribution as
\begin{align}\label{variationalpost}
    &q_{\boldsymbol{\varTheta}}(\boldsymbol{\mathcal{H}},{\{\boldsymbol{\mathcal{S}}_k\}^{N_t}_{k=1}}\vert \boldsymbol{\mathcal{Y}})\nonumber\\
    \propto{}& 
    p_{\boldsymbol{\varGamma}}(\boldsymbol{\mathcal{Y}}\vert \boldsymbol{\mathcal{H}},{\{\boldsymbol{\mathcal{S}}_k\}^{N_t}_{k=1}})q_{\boldsymbol{\theta}_{\mathcal{H}}}\left(\boldsymbol{\mathcal{H}}\right)\prod\nolimits^{N_t}_{k=1}q_{\boldsymbol{\theta}_{\mathcal{S}_k}}(\boldsymbol{\mathcal{S}}_{k})
\end{align}
where $q_{\boldsymbol{\theta}_{\mathcal{H}}}(\boldsymbol{H})$ and $\{q_{\boldsymbol{\theta}_{\mathcal{S}_k}}(\boldsymbol{\mathcal{S}}_{k})\}_{k=1}^{N_t}$ are learned priors parameterized by $\boldsymbol{\varTheta} = [\boldsymbol{\theta}_{\mathcal{H}}, \boldsymbol{\theta}_{\mathcal{S}_1},\cdots, \boldsymbol{\theta}_{\mathcal{S}_{N_t}}]$. Given that the likelihood $p_{\boldsymbol{\varGamma}}(\boldsymbol{\mathcal{Y}}\vert \boldsymbol{\mathcal{H}},{\{\boldsymbol{\mathcal{S}}_k\}^{N_t}_{k=1}})$ is fixed, the objective in \eqref{optimizeP4} is theoretically minimized when the learned priors converge to the true prior distributions, i.e., $q_{\boldsymbol{\theta}_{\mathcal{H}}}(\boldsymbol{\mathcal{H}}) \approx p(\boldsymbol{\mathcal{H}})$ and $q_{\boldsymbol{\theta}_{\mathcal{S}_k}}(\boldsymbol{\mathcal{S}}_k) \approx p(\boldsymbol{\mathcal{S}}_k)$ for all $k$.
To approximate the true priors,
we adopt the optimal transport (OT) flow-matching framework \cite{lipman2023flow} to model the generative priors. 
Specifically, we construct probability paths $p(\boldsymbol{\mathcal{H}}(\tau))$ and $p(\boldsymbol{\mathcal{S}}_{k}(\tau))$ parameterized by $\tau \in [0, 1]$, where the flow variables $\boldsymbol{\mathcal{H}}(\tau)\sim p(\boldsymbol{\mathcal{H}}(\tau))$ and $\boldsymbol{\mathcal{S}}_{k}(\tau)\sim p(\boldsymbol{\mathcal{S}}_{k}(\tau))$ evolve from pure noise at $\tau=1$ to the ground truth of the channel and the source data at $\tau=0$ along the probability paths. 
The unconditional generation is governed by the flow ODEs $\frac{d \boldsymbol{\mathcal{H}}(\tau)}{d \tau}=
\boldsymbol{\mathcal{V}}(\boldsymbol{\mathcal{H}}(\tau))$ and
$\frac{d \boldsymbol{\mathcal{S}}_k(\tau)}{d \tau}=
\boldsymbol{\mathcal{V}}(\boldsymbol{\mathcal{S}}_k(\tau))$,
where $\boldsymbol{\mathcal{V}}(\boldsymbol{\mathcal{H}}(\tau))$ and $\boldsymbol{\mathcal{V}}(\boldsymbol{\mathcal{S}}_k(\tau))$ are the unconditional prior velocity fields (VFs). 
Leveraging flow-matching techniques \cite{lipman2023flow, kim2025flowdps}, we learn the prior VFs via DNNs, and obtain the learned prior VFs $\boldsymbol{\mathcal{V}}^{\boldsymbol{\theta}_{\mathcal{H}}}$ and $\boldsymbol{\mathcal{V}}^{\boldsymbol{\theta}_{\mathcal{S}_k}}$.

\begin{remark}
Accoding to \cite[Eq. (17)]{kim2025flowdps}, the learned prior VFs are parametrized by the scores of the generative priors as
\begin{subequations}\label{eq:vel_score_relation}
\begin{align}
    \boldsymbol{\mathcal{V}}^{\boldsymbol{\theta}_{\mathcal{H}}}(\boldsymbol{\mathcal{H}}(\tau)) &= \frac{\boldsymbol{\mathcal{H}}(\tau) + \tau\nabla_{\boldsymbol{\mathcal{H}}(\tau)}\log q_{\boldsymbol{\theta}_{\mathcal{H}}}(\boldsymbol{\mathcal{H}}(\tau))}{\tau-1},\\
    \boldsymbol{\mathcal{V}}^{\boldsymbol{\theta}_{\mathcal{S}_k}}(\boldsymbol{\mathcal{S}}_k(\tau)) &= \frac{\boldsymbol{\mathcal{S}}_k(\tau) + \tau\nabla_{\boldsymbol{\mathcal{S}}_k(\tau)}\log q_{\boldsymbol{\theta}_{\mathcal{S}_k}}(\boldsymbol{\mathcal{S}}_k(\tau))}{\tau-1},
\end{align}
\end{subequations}
which bridges the gap between the implicit generative priors $q_{\boldsymbol{\theta}_{\mathcal{H}}}(\boldsymbol{\mathcal{H}})$ , $q_{\boldsymbol{\theta}_{\mathcal{S}_k}}(\boldsymbol{\mathcal{S}}_k)$ and the explicit learned prior VFs.
\end{remark}

To realize channel estimation and source recovery via sampling the variational posterior in \eqref{variationalpost}, we impose $\tau$ into the variational posterior to obtain the posterior path as
\begin{align}\label{eq:post_path_time}
    &q_{\boldsymbol{\varTheta}}(\boldsymbol{\mathcal{H}}(\tau),\{\boldsymbol{\mathcal{S}}_k(\tau)\}^{N_t}_{k=1}\vert \boldsymbol{\mathcal{Y}}) \nonumber\\\propto &p_{\boldsymbol{\varGamma}}(\boldsymbol{\mathcal{Y}}\vert \boldsymbol{\mathcal{H}}(\tau),{\{\boldsymbol{\mathcal{S}}_k(\tau)\}^{N_t}_{k=1}})
    \nonumber\\&\times
    q_{\boldsymbol{\theta}_{\mathcal{H}}}\left(\boldsymbol{\mathcal{H}}(\tau)\right)\prod\nolimits^{N_t}_{k=1}q_{\boldsymbol{\theta}_{\mathcal{S}_k}}(\boldsymbol{\mathcal{S}}_{k}(\tau)),
\end{align}
which is equivalent to the original variational posterior \eqref{variationalpost} at $\tau=0$. As $\boldsymbol{\mathcal{H}}(\tau)$ and $\{\boldsymbol{\mathcal{S}}_{k}(\tau)\}^{N_t}_{k=1}$ evolve from $\tau=1$ to $\tau=0$ along the posterior probability path in \eqref{eq:post_path_time}, we can obtain the posterior samples $\{\hat{\boldsymbol{\mathcal{H}}},{\{\hat{\boldsymbol{\mathcal{S}}}_k\}^{N_t}_{k=1}}\}$.
By substituting the generative priors in \eqref{eq:vel_score_relation} with the posterior probability paths in \eqref{eq:post_path_time}, the posterior VFs are expressed as \cite{kim2025flowdps}
\begin{subequations}\label{eq:post-vel-bayes}
\begin{align}
\boldsymbol{\mathcal{V}}^{\boldsymbol{\theta}_{\mathcal{H}}}&(\boldsymbol{\mathcal{H}}(\tau) \vert \boldsymbol{\mathcal{Y}}
, \{\boldsymbol{\mathcal{S}}_k(\tau)\}^{N_t}_{k=1}
)
=
\boldsymbol{\mathcal{V}}^{\boldsymbol{\theta}_{\mathcal{H}}}(\boldsymbol{\mathcal{H}}(\tau))\nonumber\\
+&
\frac{\tau}{\tau-1} \nabla_{\boldsymbol{\mathcal{H}}(\tau)}\log p_{\boldsymbol{\varGamma}}(\boldsymbol{\mathcal{Y}}\vert\boldsymbol{\mathcal{H}}(\tau), \{\boldsymbol{\mathcal{S}}_k(\tau)\}^{N_t}_{k=1}),\\
\boldsymbol{\mathcal{V}}^{\boldsymbol{\theta}_{\mathcal{S}_k}}&(\boldsymbol{\mathcal{S}}_k(\tau) \vert \boldsymbol{\mathcal{Y}}
,\boldsymbol{\mathcal{H}}(\tau), \{\boldsymbol{\mathcal{S}}_{k^\prime}(\tau)\}_{k^\prime\neq k}^{N_t}
)
=
\boldsymbol{\mathcal{V}}^{\boldsymbol{\theta}_{\mathcal{S}_k}}(\boldsymbol{\mathcal{S}}_k(\tau))\nonumber\\
+&
\frac{\tau}{\tau-1} \nabla_{\boldsymbol{\mathcal{S}}_k(\tau)}\log p_{\boldsymbol{\varGamma}}(\boldsymbol{\mathcal{Y}}\vert\boldsymbol{\mathcal{H}}(\tau), \{\boldsymbol{\mathcal{S}}_k(\tau)\}^{N_t}_{k=1}).
\end{align}
\end{subequations}
Given the posterior VFs, the posterior sampling is equivalent to numerically integrating the flow ODEs with the posterior VFs in \eqref{eq:post-vel-bayes} from $\tau=1$ to $\tau=0$ \cite{lipman2023flow}.
Using the Euler method \cite{strang2000linear} with a reverse step size $\Delta \tau > 0$, the update rules for the reverse integrations are given by
\begin{subequations}\label{eq:ODE_Euler}
\begin{align}
    &\boldsymbol{\mathcal{H}}(\tau-\Delta \tau)
    = \boldsymbol{\mathcal{H}}(\tau) - \Delta \tau\boldsymbol{\mathcal{V}}^{\boldsymbol{\theta}_{\mathcal{H}}}(\boldsymbol{\mathcal{H}}(\tau))\nonumber\\
    &+
    \frac{\tau\Delta\tau}{1-\tau} \nabla_{\boldsymbol{\mathcal{H}}(\tau)}\log p_{\boldsymbol{\varGamma}}(\boldsymbol{\mathcal{Y}}\vert\boldsymbol{\mathcal{H}}(\tau), \{\boldsymbol{\mathcal{S}}_k(\tau)\}^{N_t}_{k=1}),\\
    &\boldsymbol{\mathcal{S}}_k(\tau-\Delta \tau)
    =\boldsymbol{\mathcal{S}}_k(\tau) - \Delta \tau\boldsymbol{\mathcal{V}}^{\boldsymbol{\theta}_{\mathcal{S}_k}}(\boldsymbol{\mathcal{S}}_k(\tau))\nonumber\\
    &+
    \frac{\tau\Delta\tau}{1-\tau} \nabla_{\boldsymbol{\mathcal{S}}_k(\tau)}\log p_{\boldsymbol{\varGamma}}(\boldsymbol{\mathcal{Y}}\vert\boldsymbol{\mathcal{H}}(\tau), \{\boldsymbol{\mathcal{S}}_k(\tau)\}^{N_t}_{k=1}).
\end{align}
\end{subequations}

The complete update rules are summarized in Algorithm~\ref{alg:parallel-flowdps}. Note that directly computing the likelihood scores in \eqref{eq:ODE_Euler} is intractable as the likelihood is conditioned on the flow variables rather than the ground truth.
To approximately calculate the likelihood scores, we leverage Tweedie's formula \cite{chung2023diffusion} to obtain the minimum mean square error (MMSE) estimates of the ground truth based on the flow variables, which are given in step 3 of Alg.~\ref{alg:parallel-flowdps}.
Regarding the MMSE estimates as the ground truth, the likelihood scores can be approximated as in step 4 of Alg.~\ref{alg:parallel-flowdps}, where $\beta_{\mathcal{H}}$ and $\{\beta_{\mathcal{S}_k}\}^{N_t}_{k=1}$ are hyper-parameters used to compensate for approximation errors. 
Substituting the approximated likelihood scores into \eqref{eq:ODE_Euler} gives the final update rules of posterior sampling in step 5 of Alg.~\ref{alg:parallel-flowdps}.
Since all the flow variables share the same likelihood at any $\tau$, their updates can be executed in a fully parallel manner, which ensures high computational efficiency.



\begin{algorithm}[htb]
\caption{Parallel Flow-Matching Decoding}
\label{alg:parallel-flowdps}
\begin{algorithmic}[1]
\REQUIRE Receive signal $\boldsymbol{\mathcal{Y}}$; 
number of iterations $1/\Delta \tau$; 
DU encoders $\{f_{\boldsymbol{\gamma}_k}\}^{N_t}_{k=1}$;
pre-trained prior VFs $\boldsymbol{\mathcal{V}}^{\boldsymbol{\theta}_{\mathcal{H}}}$ and $\{\boldsymbol{\mathcal{V}}^{\boldsymbol{\theta}_{\mathcal{S}_k}}\}^{N_t}_{k=1}$; hyper-parameters $\beta_{\mathcal{H}}$ and $\{\beta_{\mathcal{S}_k}\}^{N_t}_{k=1}$.
\ENSURE $\hat{\boldsymbol{\mathcal{H}}}, \{\hat{\boldsymbol{\mathcal{S}}}_{k}\}_{k=1}^{N_t}$.
\STATE Sample $\boldsymbol{\mathcal{H}}(1)$ and $\{\boldsymbol{\mathcal{S}}_k(1)\}^{N_t}_{k=1}$ from standard Gaussians\\
\FOR{$\tau=1,1-\Delta \tau,\cdots,\Delta \tau$}
    \STATE \textit{MMSE Estimation via Tweedie's Formula:}\\
    $\hat{\boldsymbol{\mathcal{H}}}(0\vert \tau) = \boldsymbol{\mathcal{H}}(\tau) - \tau\boldsymbol{\mathcal{V}}^{\boldsymbol{\theta}_{\mathcal{H}}}(\boldsymbol{\mathcal{H}}(\tau))$\\
    $\hat{\boldsymbol{\mathcal{S}}}_k(0\vert \tau) = \boldsymbol{\mathcal{S}}_k(\tau) - \tau\boldsymbol{\mathcal{V}}^{\boldsymbol{\theta}_{\mathcal{S}_k}}(\boldsymbol{\mathcal{S}}_k(\tau)),\forall k$
    \STATE \textit{Likelihood Score Approximation:}\\
    $\boldsymbol{\mathcal{G}}_{\mathcal{H}} = \nabla_{\hat{\boldsymbol{\mathcal{H}}}(0\vert \tau)}\log p_{\boldsymbol{\varGamma}}(\boldsymbol{\mathcal{Y}}\vert\hat{\boldsymbol{\mathcal{H}}}(0\vert \tau), \{\hat{\boldsymbol{\mathcal{S}}}_k(0\vert \tau)\}^{N_t}_{k=1})$\\
    $\boldsymbol{\mathcal{G}}_{\mathcal{S}_k} = \nabla_{\hat{\boldsymbol{\mathcal{S}}}_k(0\vert \tau)}\log p_{\boldsymbol{\varGamma}}(\boldsymbol{\mathcal{Y}}\vert\hat{\boldsymbol{\mathcal{H}}}(0\vert \tau), \{\hat{\boldsymbol{\mathcal{S}}}_k(0\vert \tau)\}^{N_t}_{k=1})$
    \STATE \textit{Posterior Sampling via Euler Method:}\\
        $\boldsymbol{\mathcal{H}}(\tau-\Delta \tau)
    = \boldsymbol{\mathcal{H}}(\tau) - \boldsymbol{\mathcal{V}}^{\boldsymbol{\theta}_{\mathcal{H}}}(\boldsymbol{\mathcal{H}}(\tau))\Delta\tau+
    \frac{\tau\Delta\tau\beta_{\mathcal{H}}}{1-\tau} \boldsymbol{\mathcal{G}}_{\mathcal{H}}$
    $\boldsymbol{\mathcal{S}}_k(\tau-\Delta \tau) = \boldsymbol{\mathcal{S}}_k(\tau) - \boldsymbol{\mathcal{V}}^{\boldsymbol{\theta}_{\mathcal{S}_k}}(\boldsymbol{\mathcal{S}}_k(\tau))\Delta\tau+
    \frac{\tau\Delta\tau\beta_{\mathcal{S}_k}}{1-\tau}\boldsymbol{\mathcal{G}}_{\mathcal{S}_k}$
\ENDFOR
\RETURN $\hat{\boldsymbol{\mathcal{H}}} = \boldsymbol{\mathcal{H}}(\Delta\tau)$ and $\{\hat{\boldsymbol{\mathcal{S}}}_{k} = \boldsymbol{\mathcal{S}}_{k}(\Delta\tau)\}_{k=1}^{N_t}$
\end{algorithmic}
\end{algorithm}

\section{Performance Analysis}\label{sec:theory}

In this section, we analyze the fundamental limits of the joint estimation problem via BCRB.
We first derive the Fisher information matrix (FIM). Let $\boldsymbol{y} = \operatorname{vec}(\boldsymbol{\mathcal{Y}}) \in \mathbb{C}^{N_f T N_r \times 1}$, $\boldsymbol{h} = \operatorname{vec}(\boldsymbol{\mathcal{H}}) \in \mathbb{C}^{N_f N_t N_r \times 1}$, and $\boldsymbol{x} = \operatorname{vec}(\boldsymbol{\mathcal{X}}) \in \mathbb{C}^{N_f T N_t \times 1}$.
The receive signal is vectorized as $\boldsymbol{y} = (\mathop{\oplus}\nolimits^{N_f}_{f=1} (\mathbf{I}_{N_r} \otimes \boldsymbol{\mathcal{X}}_{f,:,:}))\boldsymbol{h} + \boldsymbol{n}$, where $\oplus$ denotes the direct sum of matrices,\footnote{The direct sum of matrices is a block diagonal matrix expressed as $\mathop{\oplus}\nolimits^{N_f}_{f=1} \boldsymbol{A}_f=\operatorname{Diag}[\boldsymbol{A}_1,\ldots,\boldsymbol{A}_{N_f}]$.} 
and $\otimes$ is the Kronecker product.
The likelihood $p(\boldsymbol{y} \vert \boldsymbol{x}, \boldsymbol{h})$ is a Gaussian distribution.
The FIM is defined using Wirtinger derivatives of the log-likelihood. The resulting FIM is given by
\begin{align}\label{eq:FOP_FSP}
&\boldsymbol{F}(\boldsymbol{x}, \boldsymbol{h}) = 2/\sigma_n^2\times\nonumber \\&\begin{bmatrix} \mathop{\oplus}^{N_f}_{f=1}( (\boldsymbol{\mathcal{H}}^*_{f,:,:} \boldsymbol{\mathcal{H}}^\top_{f,:,:}) \otimes \mathbf{I}_{T}) & \mathop{\oplus}^{N_f}_{f=1}(\boldsymbol{\mathcal{H}}^*_{f,:,:} \otimes \boldsymbol{\mathcal{X}}_{f,:,:}) \\ \mathop{\oplus}^{N_f}_{f=1}(\boldsymbol{\mathcal{H}}^\top_{f,:,:} \otimes \boldsymbol{\mathcal{X}}^\mathrm{H}_{f,:,:}) & \mathop{\oplus}^{N_f}_{f=1}(\mathbf{I}_{N_r} \otimes (\boldsymbol{\mathcal{X}}^\mathrm{H}_{f,:,:} \boldsymbol{\mathcal{X}}_{f,:,:})) \end{bmatrix},
\end{align}
and the CRB on the estimation error is the inverse of FIM.
However, the FIM is always rank deficient as proved in Proposition~\ref{prop:FIM_rank_deficiency}, and consequently, the CRB cannot be computed.

\begin{proposition}[Rank Deficiency of $\boldsymbol{F}(\boldsymbol{x}, \boldsymbol{h})$]\label{prop:FIM_rank_deficiency}
The FIM $\boldsymbol{F}(\boldsymbol{x}, \boldsymbol{h}) \in\mathbb{C}^{N_fN_t(T+N_r)\times N_fN_t(T+N_r)}$ is always rank-deficient with $\operatorname{rank}(\boldsymbol{F}(\boldsymbol{x}, \boldsymbol{h})) \leq N_fN_t(T+N_r) - N_fN_t^2$.
\end{proposition}

\begin{proof}
See Appendix~\ref{Appendix2}.
\end{proof}

\begin{remark}[Interpretation of Rank Deficiency]\label{Identifiability}
The rank deficiency of $\boldsymbol{F}(\boldsymbol{x}, \boldsymbol{h})$ quantified in Proposition~\ref{prop:FIM_rank_deficiency} stems from the non-uniqueness of matrix factorization \cite{jiang2024generalized, sidiropoulos2017tensor},
as the equality $\boldsymbol{\mathcal{X}}_{f,:,:}\boldsymbol{\mathcal{H}}_{f,:,:} = (\boldsymbol{\mathcal{X}}_{f,:,:}\boldsymbol{\Pi})( \boldsymbol{\Pi}^{\top}\boldsymbol{\mathcal{H}}_{f,:,:})$ holds for any permutation matrix $\boldsymbol{\Pi}$ that reorders the columns of $\boldsymbol{\mathcal{X}}_{f,:,:}$ and rows of $\boldsymbol{\mathcal{H}}_{f,:,:}$. 
Such non-uniqueness is challenging to resolve when the elements of the codeword are i.i.d., as $\boldsymbol{\mathcal{X}}_{f,:,:}$ and $\boldsymbol{\mathcal{X}}_{f,:,:}\boldsymbol{\Pi}$ follow identical distributions, thereby rendering them statistically indistinguishable.
This non-uniqueness can be resolved by exploiting the prior information. This includes:
\begin{itemize}
    \item \textbf{Deterministic pilots} \cite{mashhadi2021pruning, han2025interference}: 
    The pilots provide known reference symbols that anchor the scale, phase, and permutation of $\boldsymbol{\mathcal{X}}$, effectively acting as deterministic priors.
    \item \textbf{Statistical priors} \cite{jiang2024generalized, zhang2018blind}: In a Bayesian framework, assuming that $\boldsymbol{\mathcal{X}}$ and $\boldsymbol{\mathcal{H}}$ follows informative distributions, the Bayesian FIM can be full-rank even without pilots.
\end{itemize}
\end{remark}

To achieve a full-rank FIM, we consider the prior information of the channel and signal to obtain the BFIM and the BCRB. The BFIM is defined as $\boldsymbol{F}_B(\boldsymbol{x}, \boldsymbol{h}) = \mathbb{E}_{p(\boldsymbol{x}, \boldsymbol{h})} \left[ \boldsymbol{F}(\boldsymbol{x}, \boldsymbol{h}) \right] +  \boldsymbol{F}_{P}(\boldsymbol{x}, \boldsymbol{h})$,
where the first term is the expectation of FIM in \eqref{eq:FOP_FSP}, and the second term is the prior FIM.
Since priors of the channel and transmit signal are independent, $\boldsymbol{F}_{P}(\boldsymbol{x}, \boldsymbol{h}) = \operatorname{Diag}[\boldsymbol{F}_{P, \boldsymbol{x}},\boldsymbol{F}_{P, \boldsymbol{h}}]$ is block-diagonal, where $\boldsymbol{F}_{P, \boldsymbol{x}} = \mathbb{E}_{p(\boldsymbol{\mathcal{X}})}[\nabla_{\boldsymbol{x}^*} \log p(\boldsymbol{\mathcal{X}})\nabla_{\boldsymbol{x}^*} \log p(\boldsymbol{\mathcal{X}})^\mathrm{H}]$ and $\boldsymbol{F}_{P, \boldsymbol{h}} = \mathbb{E}_{p(\boldsymbol{\mathcal{H}})}[\nabla_{\boldsymbol{h}^*} \log p(\boldsymbol{\mathcal{H}})\nabla_{\boldsymbol{h}^*} \log p(\boldsymbol{\mathcal{H}})^\mathrm{H}]$. By incorporating the statistical priors of $\boldsymbol{\mathcal{H}}$ and $\boldsymbol{\mathcal{X}}$, the BFIM $\boldsymbol{F}_B(\boldsymbol{x}, \boldsymbol{h})$ can achieve full rank.
A valid BCRB can then be obtained via the BFIM.
However, in practice, the channel $\boldsymbol{\mathcal{H}}$ and the transmit signal $\boldsymbol{\mathcal{X}}$ typically lie on low-dimensional manifolds. Consequently, their density is zero almost everywhere~\cite{song2019generative}. 
Under such circumstances, scores $\nabla\log p(\boldsymbol{\mathcal{H}})$ and $\nabla\log p(\boldsymbol{\mathcal{X}})$ are unbounded~\cite{song2019generative}.
To circumvent this issue, we consider the slightly noised variables on the OT path, i.e., $\boldsymbol{\mathcal{H}}(\epsilon)$ and $\boldsymbol{\mathcal{X}}(\epsilon)$,
whose distributions are absolutely continuous with a small $\epsilon > 0$. Based on \eqref{eq:vel_score_relation}, we can compute the scores of $\boldsymbol{\mathcal{H}}(\epsilon)$ and $\boldsymbol{\mathcal{X}}(\epsilon)$, and the resulting BFIM inverse is ${\boldsymbol{F}}_B(\boldsymbol{x}, \boldsymbol{h})^{-1}= \begin{bmatrix} {\boldsymbol{C}}_{xx} & {\boldsymbol{C}}_{xh} \\ {\boldsymbol{C}}_{hx} & {\boldsymbol{C}}_{hh} \end{bmatrix}$.
Then the BCRBs for the NMSEs of estimates $\hat{\boldsymbol{\mathcal{H}}}$ and $\hat{\boldsymbol{\mathcal{X}}}$ are
\begin{subequations}
\begin{align}
    \text{BCRB}_{\mathcal{H}} &= \frac{\operatorname{tr}( {\boldsymbol{C}}_{hh} )}{\mathbb{E}\left[\|\boldsymbol{\mathcal{H}}\|_F^2 \right] } \leq \mathbb{E}\left[\frac{ \|\hat{\boldsymbol{\mathcal{H}}} - \boldsymbol{\mathcal{H}}\|_F^2}{\|\boldsymbol{\mathcal{H}}\|_F^2} \right], \label{eq:MSE_H} \\
    \text{BCRB}_{\mathcal{X}} &= \frac{\operatorname{tr}( {\boldsymbol{C}}_{xx} )}{\mathbb{E}\left[\|\boldsymbol{\mathcal{X}}\|_F^2 \right] } \leq \mathbb{E}\left[\frac{ \|\hat{\boldsymbol{\mathcal{X}}} - \boldsymbol{\mathcal{X}}\|_F^2}{\|\boldsymbol{\mathcal{X}}\|_F^2} \right]. \label{eq:MSE_D}
\end{align}
\end{subequations}

The approximated score functions involve two main sources of approximation errors: 1) The errors between the learned scores and the true scores of $p(\boldsymbol{\mathcal{H}}(\epsilon))$ and $p(\boldsymbol{\mathcal{X}}(\epsilon))$, which are bounded as $\mathbb{E} \left[\epsilon \| \nabla \log q_{\boldsymbol{\theta}_{\mathcal{H}}}(\boldsymbol{\mathcal{H}}(\epsilon)) - \nabla \log p(\boldsymbol{\mathcal{H}}(\epsilon)) \|_2 \right] \leq \delta_{\mathcal{H}}$ and $\mathbb{E} \left[\epsilon \| \nabla \log q_{\boldsymbol{\theta}_{\mathcal{H}}}(\boldsymbol{\mathcal{X}}(\epsilon)) - \nabla \log p(\boldsymbol{\mathcal{X}}(\epsilon)) \|_2 \right] \leq \delta_{\mathcal{X}}$,
where $\delta_{\mathcal{H}}$ and $\delta_{\mathcal{X}}$ are directly controlled by the final training loss of flow-matching.
2) The errors between the $\epsilon$-smoothed scores and the prior scores, which are bounded as
    $\left\| P_{\mathcal{T}_{\boldsymbol{h}}} ( \nabla \log p(\boldsymbol{\mathcal{H}}(\epsilon))) - P_{\mathcal{T}_{\boldsymbol{h}}}(\nabla \log p(\boldsymbol{\mathcal{H}})) \right\|_2 \leq \epsilon C_h$ and
    $\left\| P_{\mathcal{T}_{\boldsymbol{x}}}( \nabla \log p(\boldsymbol{\mathcal{X}}(\epsilon))) - P_{\mathcal{T}_{\boldsymbol{x}}}(\nabla \log p(\boldsymbol{\mathcal{X}})) \right\|_2 \leq \epsilon C_x$,
where $P_{\mathcal{T}_{\boldsymbol{h}}}$ and $P_{\mathcal{T}_{\boldsymbol{x}}}$ respectively denote the projection onto the tangent spaces of the low-dimensional channel and signal manifolds, and $C_h$ and $C_x$ are finite constants depending on the smoothness of the prior distributions. 
For detailed derivations, please refer to Appendix~\ref{Appendix:Error-Bounds}. 
As training converges ($\delta_{\mathcal{H}}, \delta_{\mathcal{X}} \to 0$) and smoothing is annealed ($\epsilon \to 0$), the approximated scores converge to the prior scores.

\section{Numerical Results}\label{Section:NR}


In simulations, we consider the UMa-NLOS scenario \cite{tr385g} for channel generation.
The receiver is equipped with a dual-polarized uniform planar array with half-wavelength antenna spacing, consisting of $4$ vertical elements and $16$ horizontal elements. Each array element comprises a co-located cross-polarized antenna pair, resulting in a total of $N_r=128$ receive antennas. The number of transmitters is $N_t=4$.
The subcarrier spacing is $30$~kHz, the number of subcarriers is $N_f=74$, the carrier frequency is $6.7$~GHz, the blocklength is $T=14$, and the Doppler shift is $15$ Hz. 
The training set consists of 50,000 samples, and the evaluation set includes 5,000 samples.
The source dataset is the FFHQ-256 dataset, which consists of 70,000 high-quality images of human faces at $256\times 256$ resolution. The training set includes 60,000 image samples, and the remaining 10,000 images are for validation.

The quality of image reconstruction is quantified through PSNR, SSIM, and LPIPS, denoted as $\text{PSNR}_{\mathcal{S}}$, $\text{SSIM}_{\mathcal{S}}$, and $\text{LPIPS}_{\mathcal{S}}$, respectively.
The channel estimation performance is measured via $\text{NMSE}_{\mathcal{H}} = \Vert  \boldsymbol{\mathcal{H}}-\hat{\boldsymbol{\mathcal{H}}}\Vert^2_F/\Vert  \boldsymbol{\mathcal{H}}\Vert^2_F$.
The NMSE of the estimated transmit signal is $\text{NMSE}_{\mathcal{X}} = \Vert  \boldsymbol{\mathcal{X}}-\hat{\boldsymbol{\mathcal{X}}}\Vert^2_F/\Vert  \boldsymbol{\mathcal{X}}\Vert^2_F$, where $\hat{\boldsymbol{\mathcal{X}}}$ is obtained by inputting $\{\hat{\boldsymbol{\mathcal{S}}}_k\}^{N_t}_{k=1}$ into the DU encoders.
The channel bandwidth ratio is defined as $\text{CBR} \triangleq \frac{N_fT_{\mathcal{S}}}{n}$ and $n=256\times 256\times 3$ is the source data dimension.
The channel quality is defined as $\text{CSNR} = \Vert  \boldsymbol{\mathcal{Y}}-\boldsymbol{\mathcal{W}}\Vert^2_F/\Vert  \boldsymbol{\mathcal{W}}\Vert^2_F$.
\begin{figure}[htb]
    \centering
    \includegraphics[width=1\linewidth]{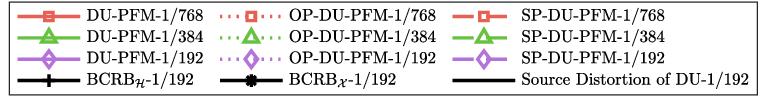}
    \includegraphics[width=0.323\linewidth]{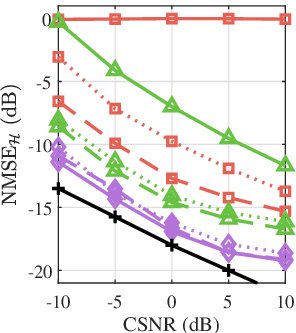}
    \includegraphics[width=0.323\linewidth]{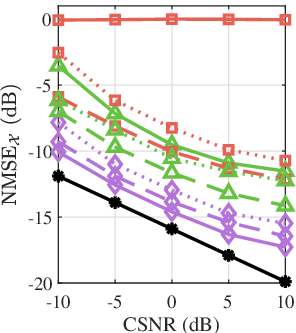}
    \includegraphics[width=0.323\linewidth]{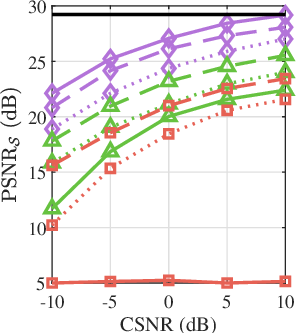}
    \caption{$\text{NMSE}_{\mathcal{H}}$, $\text{NMSE}_{\mathcal{X}}$, and $\text{PSNR}_{\mathcal{S}}$ of the DU-PFM framework under various CBRs and CSNRs. The pilot ratio for OP/SP-DU-PFM is $\alpha=0.5$.}
    \label{fig:NMSE_PSNR_vs_CBR}
\end{figure}

We first evaluate the performance of the proposed framework with a varying CBR.
Fig.~\ref{fig:NMSE_PSNR_vs_CBR} compares DU-PFM with its pilot-assisted counterparts under different CBRs and CSNRs, with $\alpha$ set to $0.5$ for OP-DU-PFM and SP-DU-PFM. 
\textbf{At a low CBR of 1/768} (corresponding to $256$ transmit symbols per image), DU-PFM fails to reliably recover either the channel or the source data.
This is because, at such a low CBR level, the DU encoder generates highly compressed codewords, where the entries of a codeword are close to statistically independent. Consequently, the limited redundancy in a codeword is insufficient to resolve the non-uniqueness issue as discussed in Remark~\ref{Identifiability}.
In this case, extra redundancy (i.e., pilots) needs to be introduced to eliminate the reconstruction uncertainty. Particularly, SP-DU-PFM exhibits a clear performance advantage over OP-DU-PFM, primarily because the latter allocates only half of the total OFDM symbols to source transmission, thereby incurring substantial intrinsic compression loss during encoding.
\textbf{At CBR=1/384}, more redundancy is introduced in a DU codeword.
Consequently, DU-PFM benefits from enhanced codeword redundancy and can alleviate both channel and source uncertainties without explicit pilots.
Nevertheless, its performance remains inferior to its pilot-assisted counterparts.
\textbf{At CBR=1/192}, DU-PFM without explicit pilots outperforms both OP-DU-PFM and SP-DU-PFM in terms of $\text{NMSE}_{\mathcal{H}}$, $\text{NMSE}_{\mathcal{X}}$, and $\text{PSNR}_{\mathcal{S}}$. This superiority stems from the fact that the pilot overhead in pilot-assisted schemes significantly reduces the available resources for source transmission, whereas DU-PFM effectively leverages the codeword redundancy that is enough to maintain high estimation accuracy.

We also include the approximated BCRBs derived in Section~\ref{sec:theory} as theoretical references in Fig.~\ref{fig:NMSE_PSNR_vs_CBR}. As observed, the gap between the approximated BCRB and the DU-PFM approach at CBR$=1/192$ widens in the high-CSNR regime. This divergence occurs because the system performance hits an error floor determined by the intrinsic compression loss of the DU encoder. As observed, the $\text{PSNR}_{\mathcal{S}}$ curve of DU-PFM converges to the horizontal black line (representing the source distortion of DU encoder at CBR~$=1/192$, i.e., the average PSNR of $\{g_{\boldsymbol{\beta}_k}(f_{\boldsymbol{\gamma}_k}(\boldsymbol{\mathcal{S}}_{k}))\}^{N_t}_{k=1}$ in \eqref{eq:DUcodec}), confirming that the performance is bottlenecked by lossy compression.

We now turn to the computational complexity analysis. For generative learning based methods, complexity is largely determined by the number of neural function evaluations (NFE).
Fig.~\ref{fig:NFE} compares the estimation performance of DU-PFM with DU-PRMP and DU-PDPS, which respectively employ diffusion-based parallel reverse mean propagation (PRMP) \cite{jiang2025blind} and parallel DPS (PDPS) \cite{chung2023parallel} for joint channel and source estimation with the proposed DU encoder, under different NFEs.
Remarkably, $\text{PSNR}_{\mathcal{S}}$ and $\text{NMSE}_{\mathcal{H}}$ achieved by DU-PFM with 50 NFEs are both better than those achieved by DU-PDPS with 1000 NFEs and DU-PRMP with 300 NFEs.
\begin{figure}[htb]
    \centering
    \includegraphics[width=0.8\linewidth]{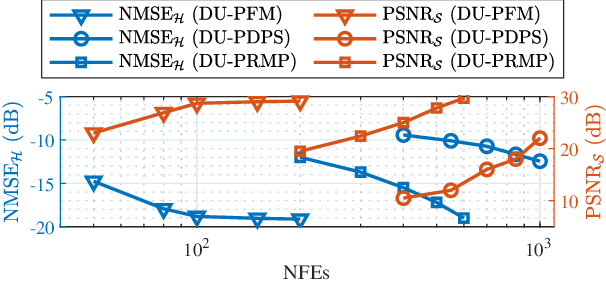}
    \caption{Tradeoff between the estimation accuracy and the computational complexity in terms of NFE, where CSNR~$=10$~dB and CBR~$=1/192$.}
    \label{fig:NFE}
\end{figure}




To provide a comprehensive evaluation, DU-PFM is compared with the following methods.
\textbf{BPG-Capacity:} An idealized reference, assuming capacity-achieving channel codes and using BPG as the source codec \cite{bellard2015bpg}; \textbf{DSC-Capacity:} The second idealized reference, assuming capacity-achieving channel codes and using a DNN-based source codec (DSC) \cite{mentzer2020high}; \textbf{BPG-LDPC:} This scheme uses the BPG codec for source compression, 5G LDPC codes with a block length of 4096 for channel coding, and modulation formats compliant with 3GPP TS 38.214 \cite{tr385g}. Pilots are used for channel estimation.
\textbf{DJSCC-MIMO:} The DJSCC-MIMO algorithm \cite{wu2024deep} with pilot-based channel estimation; \textbf{DPS-MIMO:} This scheme uses DJSCC as the codec, and source recovery is achieved via DPS \cite{chung2023diffusion} with pilot-assisted channel estimation.
Note that the pilot-based channel estimation in the last three baselines is achieved by the linear MMSE estimation with the spatial covariance of the MIMO-OFDM channel assumed to be \textit{a priori} known.
\begin{table}[htb]
\renewcommand{\arraystretch}{0.5}
\setlength{\tabcolsep}{2pt}
\centering
\caption{Quantitative results of DU-PFM and baselines at $\text{CSNR}=0$ dB. }
\label{CDL}
\begin{tabular}{l  |c c c  |c  |c}
    \toprule
    Methods
    & PSNR$_{\mathcal{S}}$& SSIM$_{\mathcal{S}}$ & LPIPS$_{\mathcal{S}}$ & NMSE$_{\mathcal{H}}$& CBR\\
    \midrule \rowcolor{gray!40}
    \textbf{DU-PFM (Ours)} & \textbf{30.05} & \textbf{0.8562}& \textbf{0.1495} & \textbf{-18.41} & \textbf{0.0234}  \\
    \midrule
    BPG-Capacity & 30.02 & 0.8496 & 0.2458 & N/A & 0.0820  \\
    DSC-Capacity & 30.03  & 0.8510 & 0.1637 & N/A & 0.0234  \\
    \midrule
    BPG-LDPC & 29.56 & 0.8343 & 0.2658 & -14.26 & 0.1025  \\
    DJSCC-MIMO & 29.71 & 0.8056 & 0.2213 & -14.22 & 0.0703  \\
    DPS-MIMO & 29.75 & 0.8399 & 0.1573 & -14.27 & 0.0469  \\
    \bottomrule
\end{tabular}
\end{table}

To ensure a fair comparison, we maintain approximately the same $\text{PSNR}_{\mathcal{S}}$ at~$30$~dB under CSNR~$=0$~dB. The quantitative results are shown in Table~\ref{CDL}.
\subsubsection{Performance Comparison with Idealized Transmission Schemes}
DU-PFM attains the target $\text{PSNR}_{\mathcal{S}}$ with a substantially lower CBR of $0.0234$, corresponding to only \textbf{28.5\%} of the CBR required by BPG-Capacity.
Compared with DSC-Capacity, DU-PFM consistently achieves a lower $\text{LPIPS}_{\mathcal{S}}$, indicating superior perceptual quality.
These gains stem from:
\begin{itemize}
\item DU encoder: 
The proposed DU encoder enables more efficient semantic feature representation than BPG.
\item PFM decoder: Our PFM decoder leverages generative priors via flow models to synthesize realistic textures, offering better perceptual quality.
\end{itemize}
\subsubsection{Performance Comparison with Classical Schemes}
The CBR of DU-PFM is only \textbf{22.8\%} of that required by BPG-LDPC.
Besides the gain from the DU encoder and the PFM decoder, the superiority of DU-PFM over standard BPG-LDPC is attributed to two additional factors:
\begin{itemize}
\item Gaussian signaling gain: Continuous Gaussian signaling avoids modulation constraints in standard schemes.
\item Joint channel and source estimation: Leveraging the data signal to assist channel estimation, DU-PFM can achieve superior $\text{NMSE}_{\mathcal{H}}$ compared to pilot-only estimation.
\end{itemize}
\subsubsection{Performance Comparison with DNN-based Schemes}
Compared to DNN-based methods, DU-PFM demonstrates higher spectral efficiency and lower $\text{NMSE}_{\mathcal{H}}$.
Notably, to compensate for channel estimation inaccuracies, DJSCC-MIMO requires significantly more redundancy in its encoder design, resulting in a CBR three times that of DU-PFM. 
DPS-MIMO leverages a powerful diffusion model-based decoder, but its performance is restricted by the separate estimation strategy.

\section{Conclusions}\label{conclusion}
We proposed a joint channel sounding and source-channel coding framework for block-fading MIMO-OFDM systems, featuring DU encoders and a PFM decoder. 
Theoretically, we derived the BCRB to benchmark the performance limits. 
Extensive simulations demonstrated the superiority of the DU-PFM approach over state-of-the-art methods across various metrics with a drastically reduced CBR and inference time.

\appendices
\section{Proof of Proposition~\ref{prop:FIM_rank_deficiency}}\label{Appendix2}
We prove the rank deficiency by explicitly constructing $N_f N_t^2$ linearly independent vectors that lie in the null space of $\boldsymbol{F}(\boldsymbol{\phi})$.
Let $\boldsymbol{\omega} \in \mathbb{C}^{N_fN_t(T+N_r)}$ be a candidate vector structured as $\boldsymbol{\omega} = [\operatorname{vec}(\boldsymbol{\mathcal{A}})^\top, \operatorname{vec}(\boldsymbol{\mathcal{B}})^\top]^\top$, where $\boldsymbol{\mathcal{A}} \in \mathbb{C}^{N_f\times T \times N_t}$ and $\boldsymbol{\mathcal{B}} \in \mathbb{C}^{N_f\times N_t \times N_r}$.
We construct a set of linearly independent vectors $\{\boldsymbol{\omega}^{(f,\kappa,\ell)}\}$ indexed by the triplet $(f, \kappa, \ell)$, where $f \in \mathcal{I}_{N_f} = \{1,\dots,N_f\}$ and $\kappa, \ell \in \mathcal{I}_{N_t}=\{1, \dots, N_t\}$. For a specific triplet, let the components of $\boldsymbol{\omega}^{(f,\kappa,\ell)}$ be defined as:
\begin{subequations}
\begin{align}\label{eq:omega}
\boldsymbol{\mathcal{A}}_{f,:,:} ={}& \boldsymbol{\mathcal{X}}_{f,:,\kappa} \boldsymbol{e}_{\ell}^\top\\
\boldsymbol{\mathcal{B}}_{f,:,:} ={}& - \boldsymbol{e}_{\kappa}\boldsymbol{\mathcal{H}}_{f,\ell,:}\\
\boldsymbol{\mathcal{A}}_{f^\prime,:,:} ={}& \mathbf{0},\quad\forall f^\prime\in \mathcal{I}_{N_f}\backslash f\\
\boldsymbol{\mathcal{B}}_{f^\prime,:,:} ={}& \mathbf{0},\quad\forall f^\prime\in \mathcal{I}_{N_f}\backslash f,
\end{align}
\end{subequations}
where $\boldsymbol{\mathcal{X}}_{f,:,\kappa}$ is the $\kappa$-th column of $\boldsymbol{\mathcal{X}}_{f,:,:}$, $\boldsymbol{\mathcal{H}}_{f,\ell,:}$ is the $\ell$-th row of $\boldsymbol{\mathcal{H}}_{f,:,:}$, and $\boldsymbol{e}_i\in\mathbb{C}^{N_t}$ denotes the standard unit vector with the $i$-th entry being 1 and all other entries being 0.
Then the FIM satisfies
\begin{align}
&\sigma_n^2\boldsymbol{F}(\boldsymbol{\phi})\boldsymbol{\omega}^{(f,\kappa,\ell)}\nonumber \\
={}& 
\begin{bmatrix} 
\operatorname{vec}(\boldsymbol{\mathcal{A}}_{1,:,:}\boldsymbol{\mathcal{H}}_{1,:,:} \boldsymbol{\mathcal{H}}^\mathrm{H}_{1,:,:}+ \boldsymbol{\mathcal{X}}_{1,:,:} \boldsymbol{\mathcal{B}}_{1,:,:} \boldsymbol{\mathcal{H}}^\mathrm{H}_{1,:,:}) \\ 
\vdots
\\
\operatorname{vec}(\boldsymbol{\mathcal{A}}_{f,:,:}\boldsymbol{\mathcal{H}}_{f,:,:} \boldsymbol{\mathcal{H}}^\mathrm{H}_{f,:,:}+ \boldsymbol{\mathcal{X}}_{f,:,:} \boldsymbol{\mathcal{B}}_{f,:,:} \boldsymbol{\mathcal{H}}^\mathrm{H}_{f,:,:}) \\ 
\vdots
\\
\operatorname{vec}(\boldsymbol{\mathcal{A}}_{N_f,:,:}\boldsymbol{\mathcal{H}}_{N_f,:,:} \boldsymbol{\mathcal{H}}^\mathrm{H}_{N_f,:,:}+ \boldsymbol{\mathcal{X}}_{N_f,:,:} \boldsymbol{\mathcal{B}}_{N_f,:,:} \boldsymbol{\mathcal{H}}^\mathrm{H}_{N_f,:,:}) \\
\operatorname{vec}(\boldsymbol{\mathcal{X}}^\mathrm{H}_{1,:,:}\boldsymbol{\mathcal{A}}_{1,:,:} \boldsymbol{\mathcal{H}}_{1,:,:}+  \boldsymbol{\mathcal{X}}^\mathrm{H}_{1,:,:} \boldsymbol{\mathcal{X}}_{1,:,:}\boldsymbol{\mathcal{B}}_{1,:,:})\\
\vdots\\
\operatorname{vec}(\boldsymbol{\mathcal{X}}^\mathrm{H}_{f,:,:}\boldsymbol{\mathcal{A}}_{f,:,:} \boldsymbol{\mathcal{H}}_{f,:,:}+  \boldsymbol{\mathcal{X}}^\mathrm{H}_{f,:,:} \boldsymbol{\mathcal{X}}_{f,:,:}\boldsymbol{\mathcal{B}}_{f,:,:})\\
\vdots\\
\operatorname{vec}(\boldsymbol{\mathcal{X}}^\mathrm{H}_{N_f,:,:}\boldsymbol{\mathcal{A}}_{N_f,:,:} \boldsymbol{\mathcal{H}}_{N_f,:,:}+  \boldsymbol{\mathcal{X}}^\mathrm{H}_{N_f,:,:} \boldsymbol{\mathcal{X}}_{N_f,:,:}\boldsymbol{\mathcal{B}}_{N_f,:,:})
\end{bmatrix}\nonumber\\
={}&
\begin{bmatrix} 
\mathbf{0} \\ 
\vdots
\\
\operatorname{vec}\big( (\boldsymbol{\mathcal{X}}_{f,:,\kappa} \boldsymbol{\mathcal{H}}_{f,\ell,:} - \boldsymbol{\mathcal{X}}_{f,:,\kappa} \boldsymbol{\mathcal{H}}_{f,\ell,:}) \boldsymbol{\mathcal{H}}^{\mathrm{H}}_{f,:,:} \big) \\ 
\vdots
\\
\mathbf{0} \\
\mathbf{0}\\
\vdots\\
\operatorname{vec}\big( \boldsymbol{\mathcal{X}}^{\mathrm{H}}_{f,:,:}(\boldsymbol{\mathcal{X}}_{f,:,\kappa} \boldsymbol{\mathcal{H}}_{f,\ell,:} - \boldsymbol{\mathcal{X}}_{f,:,\kappa} \boldsymbol{\mathcal{H}}_{f,\ell,:}) \big)\\
\vdots\\
\mathbf{0}
\end{bmatrix}\nonumber\\
={}& \mathbf{0}.
\end{align}
Finally, we count the number of such vectors. There are $N_f$ choices for $f$, and $N_t^2$ choices for the pair $(\kappa, \ell)$, yielding a total of $N_f N_t^2$ vectors.
These vectors are linearly independent because each $\boldsymbol{\omega}^{(f,\kappa,\ell)}$ is constructed using orthogonal standard basis vectors $\boldsymbol{e}_{\ell}$ and $\boldsymbol{e}_{\kappa}$, provided that $\boldsymbol{\mathcal{X}}$ and $\boldsymbol{\mathcal{H}}$ do not contain all-zero columns or rows (a condition satisfied almost surely in practical MIMO-OFDM scenarios).
Consequently, the nullity of $\boldsymbol{F}(\boldsymbol{\phi})$ is at least $N_fN_t^2$, which implies that the rank of $\boldsymbol{F}(\boldsymbol{\phi})$ is upper bounded by its total dimension $N_fN_t(T+N_r)$ minus $N_f N_t^2$.

\section{Derivation of the errors between the $\epsilon$-smoothed scores and the true prior scores}
\label{Appendix:Error-Bounds}

To rigorously quantify the approximation error between the $\epsilon$-smoothed scores and the true prior scores, we explicitly account for the low-dimensional geometric structure of the priors.
We assume that the prior distributions of the channel and signal are supported on low-dimensional smooth manifolds, denoted by $\mathcal M_h$ and $\mathcal M_x$, respectively.

\textit{a) Tangent-space projection.}
Since the prior scores are well-defined strictly within the tangent spaces of low-dimensional manifolds, we restrict our error analysis to the scores projected to the tangent spaces \cite{song2019generative}.
Let $\mathcal{T}_{\boldsymbol{h}}$ and $\mathcal{T}_{\boldsymbol{x}}$ denote the tangent spaces of $\mathcal M_h$ and $\mathcal M_x$ at the ground truth points $\boldsymbol{\mathcal{H}}\in\mathcal M_h$ and $\boldsymbol{\mathcal{X}}\in\mathcal M_x$.
Let $P_{\mathcal{T}_{\boldsymbol{h}}}$ and $P_{\mathcal{T}_{\boldsymbol{x}}}$ be the corresponding orthogonal projection operators onto these tangent spaces.

\textit{b) Lipschitz regularity of the scores.}
We assume that the projected score functions are locally Lipschitz continuous in a neighborhood of the ground truth \cite{song2019generative}. Specifically, there exist Lipschitz constants $L_h, L_x > 0$ such that:
\begin{subequations}\label{eq:Append-nabla}
\begin{align}
&\left\|
P_{\mathcal{T}_{\boldsymbol{h}}}
\left(
\nabla \log p(\boldsymbol{\mathcal{H}}(\epsilon)))
-
P_{\mathcal{T}_{\boldsymbol{h}}}(\nabla \log p(\boldsymbol{\mathcal{H}})
\right)
\right\|_2 
\nonumber\\\le{}&
L_h
\big\|
\boldsymbol{\mathcal{H}}(\epsilon)-\boldsymbol{\mathcal{H}}
\big\|_2 \\
={}&
L_h
\big\|
\epsilon\boldsymbol{\mathcal{H}}(1)-\epsilon\boldsymbol{\mathcal{H}}
\big\|_2 \label{OTApp1}\\
={}&
\epsilon C_h, 
\label{eq:Append-nabla-h}\\
&\left\|
P_{\mathcal{T}_{\boldsymbol{x}}}
\left(
\nabla \log p(\boldsymbol{\mathcal{X}}(\epsilon)))
-
P_{\mathcal{T}_{\boldsymbol{x}}}(\nabla \log p(\boldsymbol{\mathcal{X}})
\right)
\right\|_2 \nonumber\\
\le{}&
L_x
\big\|
\boldsymbol{\mathcal{X}}(\epsilon) - \boldsymbol{\mathcal{X}}
\big\|_2\\
={}&
L_x
\big\|
\epsilon\boldsymbol{\mathcal{X}}(1)-\epsilon\boldsymbol{\mathcal{X}}
\big\|_2 \label{OTApp2}\\
={}&
\epsilon C_x,
\label{eq:Append-nabla-x}
\end{align}
\end{subequations}
where $C_h = L_h
\big\|
\boldsymbol{\mathcal{H}}(1)-\boldsymbol{\mathcal{H}}
\big\|_2$ and
$C_x = L_x \|\boldsymbol{\mathcal{X}}(1)-\boldsymbol{\mathcal{X}}\|_2$ are finite constants depending on the Lipschitz constants as well as the differences between the ground truth ($\boldsymbol{\mathcal{H}}$ and $\boldsymbol{\mathcal{X}}$) and the initial flow variables ($\boldsymbol{\mathcal{H}}(1)$ and $\boldsymbol{\mathcal{X}}(1)$). Equations \eqref{OTApp1} and \eqref{OTApp2} stem from the definition of the flow variables.
Consequently, the approximation errors between the $\epsilon$-smoothed scores and the true prior scores vanish as $\epsilon\to 0$.

\bibliographystyle{IEEEtran}
\bibliography{PFM.bib}

\begin{thebibliography}{10}
\providecommand{\url}[1]{#1}
\csname url@samestyle\endcsname
\providecommand{\newblock}{\relax}
\providecommand{\bibinfo}[2]{#2}
\providecommand{\BIBentrySTDinterwordspacing}{\spaceskip=0pt\relax}
\providecommand{\BIBentryALTinterwordstretchfactor}{4}
\providecommand{\BIBentryALTinterwordspacing}{\spaceskip=\fontdimen2\font plus
\BIBentryALTinterwordstretchfactor\fontdimen3\font minus \fontdimen4\font\relax}
\providecommand{\BIBforeignlanguage}[2]{{%
\expandafter\ifx\csname l@#1\endcsname\relax
\typeout{** WARNING: IEEEtran.bst: No hyphenation pattern has been}%
\typeout{** loaded for the language `#1'. Using the pattern for}%
\typeout{** the default language instead.}%
\else
\language=\csname l@#1\endcsname
\fi
#2}}
\providecommand{\BIBdecl}{\relax}
\BIBdecl

\bibitem{wu2023fusion}
T.~Wu, Z.~Chen, M.~Tao, B.~Xia, and W.~Zhang, ``Fusion-based multi-user semantic communications for wireless image transmission over degraded broadcast channels,'' in \emph{GLOBECOM 2023-2023 IEEE Global Communications Conference}.\hskip 1em plus 0.5em minus 0.4em\relax IEEE, 2023, pp. 7623--7628.

\bibitem{wu2024deep}
H.~Wu, Y.~Shao, C.~Bian, K.~Mikolajczyk, and D.~G{\"u}nd{\"u}z, ``Deep joint source-channel coding for adaptive image transmission over {MIMO} channels,'' \emph{IEEE Trans. Wireless Commun.}, 2024.

\bibitem{gunduz2024joint}
D.~G{\"u}nd{\"u}z, M.~A. Wigger, T.-Y. Tung, P.~Zhang, and Y.~Xiao, ``Joint source--channel coding: Fundamentals and recent progress in practical designs,'' \emph{Proc. IEEE}, 2024.

\bibitem{yang2024swinjscc}
K.~Yang, S.~Wang, J.~Dai, X.~Qin, K.~Niu, and P.~Zhang, ``{SwinJSCC}: Taming swin transformer for deep joint source-channel coding,'' \emph{IEEE Trans. Cogn. Commun. Netw.}, 2024.

\bibitem{xu2025tdjscc}
M.~Xu, C.-T. Lam, Y.~Liang, B.~Ng, X.~Yuan, and S.-K. Im, ``{TDJSCC}: Low complexity and bandwidth efficient deep joint source-channel coding with {OFDM},'' \emph{IEEE Access}, 2025.

\bibitem{3gpp.38.201}
3GPP, ``{NR; Physical layer; General description},'' {3rd Generation Partnership Project (3GPP)}, TS 38.201, version 18.4.0, 09 2024.

\bibitem{deng2020ieee}
C.~Deng, X.~Fang, X.~Han, X.~Wang, L.~Yan, R.~He, Y.~Long, and Y.~Guo, ``{IEEE} 802.11 be {Wi-Fi} 7: New challenges and opportunities,'' \emph{IEEE Commun. Surv. Tutorials}, vol.~22, no.~4, pp. 2136--2166, 2020.

\bibitem{tr385g}
3GPP, ``{5G}; study on channel model for frequencies from 0.5 to 100 {GHz},'' {3rd Generation Partnership Project (3GPP)}, Tech. Rep., 2020.

\bibitem{rusek2012scaling}
F.~Rusek, D.~Persson, B.~K. Lau, E.~G. Larsson, T.~L. Marzetta, O.~Edfors, and F.~Tufvesson, ``Scaling up {MIMO}: Opportunities and challenges with very large arrays,'' \emph{IEEE Signal Process. Mag.}, vol.~30, no.~1, pp. 40--60, 2012.

\bibitem{kuai2019structured}
X.~Kuai, L.~Chen, X.~Yuan, and A.~Liu, ``Structured turbo compressed sensing for downlink massive {MIMO-OFDM} channel estimation,'' \emph{IEEE Trans. Wireless Commun.}, vol.~18, no.~8, pp. 3813--3826, 2019.

\bibitem{mashhadi2021pruning}
M.~B. Mashhadi and D.~G{\"u}nd{\"u}z, ``Pruning the pilots: Deep learning-based pilot design and channel estimation for {MIMO-OFDM} systems,'' \emph{IEEE Trans. Wireless Commun.}, vol.~20, no.~10, pp. 6315--6328, 2021.

\bibitem{han2025interference}
X.~Han, T.~Wenqiang, J.~Shi, L.~Wendong, S.~Jia, S.~Zhihua, and Z.~Zhi, ``Interference cancellation based neural receiver for superimposed pilot in multi-layer transmission,'' \emph{China Communications}, vol.~22, no.~1, pp. 75--88, 2025.

\bibitem{lipman2023flow}
Y.~Lipman, R.~T.~Q. Chen, H.~Ben-Hamu, M.~Nickel, and M.~Le, ``Flow matching for generative modeling,'' in \emph{Proc. Int. Conf. Learn. Repr. (ICLR)}, 2023.

\bibitem{tse2005fundamentals}
D.~Tse and P.~Viswanath, \emph{Fundamentals of wireless communication}.\hskip 1em plus 0.5em minus 0.4em\relax Cambridge university press, 2005.

\bibitem{cai2025end}
C.~Cai, X.~Yuan, and Y.-J.~A. Zhang, ``End-to-end learning for task-oriented semantic communications over {MIMO} channels: An information-theoretic framework,'' \emph{IEEE J. Sel. Areas Commun.}, vol.~43, no.~4, pp. 1292--1307, 2025.

\bibitem{jiang2025blind}
H.~Jiang, X.~Yuan, Y.~Huang, and Q.~Guo, ``Blind {MIMO} semantic communication via parallel variational diffusion: A completely pilot-free approach,'' \emph{arXiv preprint arXiv:2510.27043}, 2025.

\bibitem{liu2021swin}
Z.~Liu, Y.~Lin, Y.~Cao, H.~Hu, Y.~Wei, Z.~Zhang, S.~Lin, and B.~Guo, ``Swin transformer: Hierarchical vision transformer using shifted windows,'' in \emph{Proc. IEEE Int. Conf. Comput. Vis. (ICCV)}, 2021, pp. 10\,012--10\,022.

\bibitem{kim2025flowdps}
J.~Kim, B.~S. Kim, and J.~C. Ye, ``{FlowDPS}: Flow-driven posterior sampling for inverse problems,'' \emph{arXiv preprint arXiv:2503.08136}, 2025.

\bibitem{strang2000linear}
G.~Strang, \emph{Linear Algebra and Its Applications}.\hskip 1em plus 0.5em minus 0.4em\relax Thomson, Brooks/Cole, 2006.

\bibitem{chung2023diffusion}
H.~Chung, J.~Kim, M.~T. Mccann, M.~L. Klasky, and J.~C. Ye, ``Diffusion posterior sampling for general noisy inverse problems,'' in \emph{Proc. Int. Conf. Learn. Repr. (ICLR)}, 2023.

\bibitem{jiang2024generalized}
H.~Jiang, X.~Yuan, and Q.~Guo, ``Generalized bilinear factorization via hybrid vector message passing,'' \emph{IEEE Trans. Signal Process.}, vol.~72, pp. 5675--5690, 2024.

\bibitem{sidiropoulos2017tensor}
N.~D. Sidiropoulos, L.~De~Lathauwer, X.~Fu, K.~Huang, E.~E. Papalexakis, and C.~Faloutsos, ``Tensor decomposition for signal processing and machine learning,'' \emph{IEEE Trans. Signal Process.}, vol.~65, no.~13, pp. 3551--3582, 2017.

\bibitem{zhang2018blind}
J.~Zhang, X.~Yuan, and Y.-J.~A. Zhang, ``Blind signal detection in massive {MIMO}: Exploiting the channel sparsity,'' \emph{IEEE Trans. Commun.}, vol.~66, no.~2, pp. 700--712, 2017.

\bibitem{song2019generative}
Y.~Song and S.~Ermon, ``Generative modeling by estimating gradients of the data distribution,'' \emph{Adv. Neural Inf. Process. Syst. (NeurIPS)}, vol.~32, 2019.

\bibitem{chung2023parallel}
H.~Chung, J.~Kim, S.~Kim, and J.~C. Ye, ``Parallel diffusion models of operator and image for blind inverse problems,'' \emph{Proc. IEEE/CVF Conf. Comput. Vis. Pattern Recognit. (CVPR)}, 2023.

\bibitem{bellard2015bpg}
F.~Bellard, ``{BPG} image format,'' \emph{URL https://bellard.org/bpg}, 2015.

\bibitem{mentzer2020high}
F.~Mentzer, G.~D. Toderici, M.~Tschannen, and E.~Agustsson, ``High-fidelity generative image compression,'' \emph{Adv. Neural Inf. Process. Syst. (NeurIPS)}, vol.~33, pp. 11\,913--11\,924, 2020.

\end{thebibliography}

\end{document}